\RequirePackage{fix-cm}
\RequirePackage{fixltx2e}
\documentclass[11pt]{article}
\usepackage[T1]{fontenc}
\usepackage[latin9]{inputenc}
\usepackage{enumitem}
\usepackage{amsmath}
\usepackage{amsthm}
\usepackage{amssymb}
\usepackage{setspace}
\usepackage[authoryear]{natbib}
\usepackage{microtype}
\makeatletter
\numberwithin{equation}{section}
\theoremstyle{plain}
\newtheorem{lem}{\protect\lemmaname}
\theoremstyle{plain}
\newtheorem{cor}{\protect\corollaryname}
\theoremstyle{plain}
\newtheorem{thm}{\protect\theoremname}

\usepackage[a4paper]{geometry}

\usepackage{bbm}

\usepackage{subfigure}

\usepackage{tikz}

\usepackage{lmodern}
\usepackage[T1]{fontenc}

\usepackage[hidelinks]{hyperref}

\usepackage{natbib}

\bibpunct{[}{]}{;}{y}{}{,}

\renewcommand{\phi}{\varphi}

\DeclareMathOperator*{\essinf}{ess\,inf}

\theoremstyle{remark}

\usepackage[multiple]{footmisc}

\makeatother

\usepackage{babel}
\providecommand{\corollaryname}{Corollary}
\providecommand{\lemmaname}{Lemma}
\providecommand{\theoremname}{Theorem}

\begin{document}
\title{\textbf{Optimal Auction Design under Costly Learning}}
\author{Kemal Ozbek\thanks{Department of Economics, University of Southampton, University Road,
Southampton, S017 1BJ, United Kingdom. Email: mkemalozbek@gmail.com}}
\maketitle
\begin{abstract}
We study optimal auction design in an independent private values environment
where bidders can endogenously---but at a cost---improve information
about their own valuations. The optimal mechanism is two-stage: at
stage-$1$ bidders register an information acquisition plan and pay
a transfer; at stage-$2$ they bid, and allocation and payments are
determined. We show that the revenue-optimal stage-$2$ rule is the
Vickrey--Clarke--Groves (VCG) mechanism, while stage-$1$ transfers
implement the optimal screening of types and absorb information rents
consistent with incentive compatibility and participation. By committing
to VCG ex post, the pre-auction information game becomes a potential
game, so equilibrium information choices maximize expected welfare;
the stage-$1$ fee schedule then transfers an optimal amount of payoff
without conditioning on unverifiable cost scales. The design is robust
to asymmetric primitives and accommodates a wide range of information
technologies, providing a simple implementation that unifies efficiency
and optimal revenue in environments with endogenous information acquisition.
\end{abstract}

\section{Introduction}

This paper studies the design of auctions for selling a single indivisible
good under costly learning to a finite set of risk-neutral bidders
with independent private values (IPV). Two benchmark objectives organize
much of the modern theory. A welfare-maximizing mechanism seeks to
allocate the good to the participant who values it most whenever trade
is feasible, thereby implementing an efficient assignment (\citet{Vickrey61};
\citet{Clarke71}; \citet{Groves73}). A revenue-maximizing mechanism
instead chooses rules that maximize the seller\textquoteright s expected
gains subject to incentive compatibility and participation constraints
(\citet{Myerson81}). Although these problems are posed in the same
informational environment, they typically prescribe different allocation
rules and, crucially, they need not be implemented by the same auction
format---a distinction emphasized throughout the auction design literature
(\citet{Myerson81}; \citet{RS81}).

A central theme of the present paper is that another, often under-appreciated
source of divergence between revenue and welfare may arise when information
acquisition is endogenous. In many applications, bidders do not begin
with perfectly precise knowledge of their own values; rather, they
can invest in information---due diligence, inspections, expert reports---before
bidding. The auction format that the seller commits to use after any
learning has taken place shapes the marginal value of information
for each participant. Formats that make a bidder\textquoteright s
allocation probability and payment more sensitive to her realized
private information tend to induce more information acquisition; formats
that dampen this sensitivity tend to induce less (\citet{Persico00}).
Because equilibrium information choices shift the distribution of
bids and the composition of winners, they influence both the revenue
that can be extracted and the welfare that can be realized (\citet{BV02};
\citet*{CSZ04}).

In our setting, bidders can flexibly---but at a cost---improve what
they know about their own valuations before bidding. We even allow
them, in principle, to shift the mean of their belief at an additional
cost; although this option is not exercised in equilibrium, keeping
it available delivers a richer and more robust modeling environment.
We study two-stage mechanisms: at stage-$1$, each bidder pays a registration
fee and commits to an experiment (a signal technology); at stage-$2$,
bidders submit bids in the auction. We work with a direct-revelation
formulation in which a bidder reports at stage-$1$ her primitive
type (a cost-scale parameter and a mean-value parameter) to determine
the experiment she may run and the fee she pays, and at stage-$2$
reports her realized valuation to determine allocation and payment.
Types are ex-ante random with a commonly known distribution, and equilibrium
in experiment choices is Bayesian Nash. We first assume that registered
experiments are verifiable by the seller and later relax this assumption,
showing that the main results survive under weaker verification.

Our main result shows that, unlike in static models without learning,
revenue maximization and welfare maximization coincide in the optimal
two-stage design we characterize. The stage-$2$ allocation rule is
VCG---delivering ex-post efficiency and truthful bidding---while
stage-$1$ lets bidders choose any experiment compatible with their
reported types after paying the corresponding fees. Stage-$1$ incentive-compatibility
and individual-rationality pin down a fee schedule that is independent
of unverifiable cost-scale parameters yet robustly dependent on verifiable
mean-value parameters via the registered experiments. The rationale
for leaving experiment choice unconstrained is that, once stage-$2$
is VCG, the pre-auction information game is an exact potential game
(\citet{MS96}); therefore, any maximizer of the potential---here,
expected welfare---coincides with a Bayesian Nash equilibrium of
the bidders\textquoteright{} experiment-choice game.

The literature identifies two complementary channels on the divergence
of revenue and welfare maximization objectives. In the first, bidders
choose how much to learn before bidding. Here, the format\textquoteright s
payoff sensitivity to private signals determines the equilibrium investment
in information; when expected payoffs respond sharply to incremental
precision---because the probability of winning or the effective price
is highly sensitive to the realized signal---participants optimally
acquire more information (\citet{Persico00}). This matters for welfare
because precision improves allocative matching; it matters for revenue
because sharper selection can raise the expected virtual surplus that
the seller can extract. In the second channel, the seller chooses
what bidders learn by designing the information structure itself---deciding
how informative pre-auction disclosures or signals will be---jointly
with the auction rule. This transforms information release into a
policy lever for revenue and can, in general, steer the allocation
away from efficiency (\citet{BP07}).\footnote{Relatedly, classic analyses of interdependent or affiliated signals
show how information release and format interact to shape both efficiency
and revenue, even beyond IPV (\citet{MW82}).}

These forces qualify standard format comparisons. In the symmetric
IPV benchmark without endogenous learning, a first-price auction and
a second-price auction with a common reserve implement the same allocation
rule and thus are revenue-equivalent. With endogenous information
acquisition, however, the allocation rule becomes format-dependent,
precisely because the stage-$2$ rule changes bidders\textquoteright{}
learning incentives. Two formats that were revenue-equivalent absent
learning can cease to be equivalent once learning is endogenous. From
a welfare perspective, formats and transfer rules that internalize
the social value of information---so that each bidder\textquoteright s
private incentive to learn matches the societal benefit of more precise
allocation---are desirable; in private-value environments, VCG-type
mechanisms can implement both efficient learning and efficient allocation
under natural conditions (\citet{BV02}). From a revenue perspective,
the seller\textquoteright s optimal policy typically co-determines
reserves, disclosure, and, when feasible, multi-stage procedures that
govern when and how information is produced; optimal reserve policies
in this setting generally differ from the no-learning Myerson benchmark
because they must account for the equilibrium response of information
acquisition (\citet{BP07}; \citet{Shi12}; \citet*{GMSZ21}).

In sum, even within the IPV paradigm, welfare- and revenue-maximizing
designs can diverge for structural reasons---exclusion and asymmetry
in the Myerson framework---and may diverge for strategic reasons
once information acquisition is considered. Viewing the stage-$2$
auction format as a commitment to stage-$1$ learning incentives makes
clear why the designer\textquoteright s objective, the induced information
environment, and the implementable allocation rule are tightly intertwined.
Our results nevertheless show that, when registered experiments are
verifiable, a single two-stage mechanism can align these forces and
simultaneously achieve both welfare and revenue objectives. These
findings have practical and theoretical implications. Practically,
when pre-auction due diligence matters (e.g., spectrum, mineral rights,
complex procurement), formats should be judged not only by static
allocation rules but by the learning incentives they create; committing
to VCG at stage-$2$, paired with a transparent, experiment-based
registration fee, can deliver both high revenue and efficiency. Theoretically,
we refine revenue equivalence: formats that are revenue-equivalent
without learning need not remain so once learning is endogenous, but
a suitable two-stage commitment can restore alignment by making the
induced allocation rule format-invariant through efficient information
acquisition.

The paper is organized as follows. Section $2$ presents the framework,
notation, and the incentive and participation constraints, and states
the objective functions. Section $3$ develops the main analysis,
characterizing and comparing the revenue- and welfare-maximizing mechanisms.
Section $4$ extends the model to allow (costly) auditing of registered
experiments and cost parameters and examines the robustness of the
results. Section $5$ provides a discussion of the related literature
and highlights our contributions to the research area of optimal auction
design under costly learning. Section $6$ concludes.

\section{Framework}

We adapt the framework considered in \citet{Myerson81} to study the
optimal auction design when buyers can flexibly learn. 

\paragraph*{Environment:}

There is one seller who has a single object to sell. Let $N=\{i_{1},\ldots,i_{n}\}$
denote the set of bidders with $n\geq2$. We assume that the seller
and the bidders are risk neutral. Bidders have types $\theta_{i}=(r_{i},s_{i})$
drawn from the space $\Theta_{i}=[0,1]^{2}$ equipped with the Borel
sigma-algebra. We assume that bidder types $\theta_{i}$ are drawn
i.i.d. according to the distribution $F_{rs}=F_{r}\times F_{s}$ where
$F_{r}$ and $F_{s}$ are measurable c.d.f.'s with continuous density
functions $f_{r}$ and $f_{s}$ over the space $[0,1]$, respectively.
Let $\Theta=\Theta_{1}\times...\times\Theta_{n}$ denote the set of
type profiles including every bidder and let $\Theta_{-i}$ denote
the set of type profiles $\Theta_{1}\times...\times\Theta_{i-1}\times\Theta_{i+1}\times...\times\Theta_{n}$
excluding bidder $i$. Let $\theta=(r,s)$ denote a generic profile
of types $(r_{1},s_{1};...;r_{n},s_{n})\in\Theta$ and similarly,
let $\theta_{-i}=(r_{-i},s_{-i})$ denote the vector of realized types
$(r_{1},s_{1};...;r_{i-1},s_{i-1};r_{i+1},s_{i+1};...;r_{n},s_{n})\in\Theta_{-i}$
excluding bidder $i$. We sometimes use the notation $\theta_{i1}$
to indicate $r_{i}$ and $\theta_{i2}$ to indicate $s_{i}$ when
$\theta_{i}=(r_{i},s_{i})$. Let $F_{rs}^{n}(\hat{r},\hat{s})=(F_{r}\times F_{s})^{n}(\hat{r},\hat{s})$
denote the joint distribution $\prod_{i\in N}F_{r}(\hat{r}_{i})\times F_{s}(\hat{s}_{i})$
and let $F_{rs}^{n-1}(\hat{r}_{-i},\hat{s}_{-i})=(F_{r}\times F_{s})^{n-1}(\hat{r}_{-i},\hat{s}_{-i})$
denote the joint distribution $\prod_{j\neq i}F_{r}(\hat{r}_{j})\times F_{s}(\hat{s}_{j})$
excluding bidder $i\in N$. 

For each $\alpha\in[0,1]$, let $z_{\alpha}=(a+\alpha.(b-a))$ denote
the mean value over the interval $[a,b]$ where $a>0$ is the minimum
and $b>a$ is the maximum level each individual is willing to pay
for the object. Let $D=\Delta([a,b])$ denote the set of probability
density functions over $[a,b]$ and assume that $D$ is equipped with
the weak topology. We call each distribution $f\in D$ an ``experiment''.
For each $\alpha\in[0,1]$, let $D(\alpha)=\{f\in\Delta([a,b]):E_{f}(z)=\int_{[a,b]}z.f(z).dz=z_{\alpha}\}$
be the set of experiments with center (i.e., mean value) $\mu(f)$
equal to $z_{\alpha}$. Note that each $D(\alpha)$ is closed in
$D$ and we have $D=\cup_{\alpha\in[0,1]}D(\alpha)$. Let $T=[a,b]^{n}$
and for each $i\in N$, let $T_{i}=[a,b]$ and $T_{-i}=[a,b]^{n-1}$.

For each type $\theta_{i}=(r_{i},s_{i})\in\Theta_{i}$, let $c_{\theta_{i}}:D\to\mathbb{R}_{+}$
denote an information processing cost function. We assume that each
$c_{\theta_{i}}$ satisfies the following properties: (i) $c_{\theta_{i}}(.)$
is increasing in $\theta_{i1}$, (ii) $c_{\theta_{i}}(\delta_{z})=0$
if $z=z_{\theta_{i2}}$, (iii) $c_{\theta_{i}}(.)$ is increasing
in convex-order, (iv) $c_{\theta_{i}}(.)$ is convex, and (v) $c_{\theta_{i}}(.)$
is lower semi-continuous (l.s.c.). A feasible example of the abstract
cost function $c_{\theta_{i}}$ can be given as $\hat{c}_{\theta_{i}}(f)=r_{i}\int k(|z_{i}-s_{i}|)f(z_{i})dz_{i}$
for all $f\in D$ for some (strictly) increasing convex function $k:\mathbb{R}_{+}\to\mathbb{R}_{+}$
with $k(0)=0$. As such, while $r_{i}$ can be interpreted as a cost-scale,
$s_{i}$ can be interpreted as a mean-value parameter.\footnote{The idea is that each experiment $f\in D$ determines the likelihood
$f(z)$ of each class of posterior distributions with mean $z\in[a,b]$
where $z_{s_{i}}$ denotes the mean of the initial belief of a bidder
with type $\theta_{i}=(r_{i},s_{i})$. Bidders are allowed to consider
experiments $f_{i}\in D$ with centers $\mu(f_{i})$ such that $\mu(f_{i})\neq z_{s_{i}}=z_{\theta_{i2}}$.
Moreover, depending on the particular cost function, the cost of a
degenerate experiment (i.e., $\delta_{z}$ for some $z\in[a,b]$)
may or may not be zero. For instance, while $\hat{c}_{\theta_{i}}(f)\neq0$
for any $f$ with $\mu(f)\neq z_{\theta_{i2}}$, for the cost function
$\bar{c}_{\theta_{i}}(f)=r_{i}\int k(|z_{i}-\mu(f)|)f(z_{i})dz_{i}$,
we have $\bar{c}_{\theta_{i}}(\delta_{z})=0$ for any $z\in[a,b]$. }

\paragraph*{Mechanism:}

The seller uses a two-stage direct-revelation mechanism with a reservation
price to sell the good. For simplicity, assume that the object yields
no value to the seller, and so the seller optimally sets the reservation
price equal to $0$. At stage-$1$, the seller offers to each bidder
$i$ a menu of contracts $(\sigma_{i},\tau_{i})$ such that $\sigma_{i}:\Theta_{i}\to D$
assigns an experiment $\sigma_{i}(\theta_{i})\in D(\theta_{i2})$
and $\tau_{i}:\Theta_{i}\to\mathbb{R}$ charges a fee to each type
$\theta_{i}=(r_{i},s_{i})$. Each bidder pays the fee $\tau_{i}(\theta_{i})$
to conduct the experiment $\sigma_{i}(\theta_{i})$ after revealing
her type $\theta_{i}$ to the seller. The seller is assumed to be
able to confirm that the experiment $\sigma_{i}(\theta_{i})$ is conducted
once a type $\theta_{i}\in\Theta_{i}$ is registered. Later we relax
this assumption. 

At stage-$2$, each bidder observes a value realization $t_{i}\in[a,b]$
within the support of her experiment $\sigma_{i}(\theta_{i})$ and
submits a bid $t'_{i}\in[a,b]$ where $t'_{i}$ is not necessarily
equal to $t_{i}$. The seller then allocates the object to the buyers
according to the vector of probabilities $p(t',\sigma(\theta))=(p_{1}(t',\sigma(\theta)),...,p_{n}(t',\sigma(\theta)))\geq0$
such that $\sum_{i\in N}p_{i}(t',\sigma(\theta))=1$ together with
$[x_{1}(t',\sigma(\theta)),..,x_{n}(t',\sigma(\theta))]\in\mathbb{R}^{n}$
as the vector of payments $x(t',\sigma(\theta))$ where $t'=(t'_{1},...,t'_{n})\in[a,b]^{n}$
is the vector of bids submitted by the buyers at stage-$2$ and $\sigma(\theta)=(\sigma_{1}(\theta_{1}),...,\sigma_{n}(\theta_{n}))$
is the vector of experiments registered by the buyers at stage-$1$.
Let $(\sigma_{-i},\tau_{-i})=(\sigma_{1},\tau_{1};...,\sigma_{i-1},\tau_{i-1};\sigma_{i+1},\tau_{i+1};...;\sigma_{n},\tau_{n})$
denote the vector of menus of contracts given to the bidders excluding
buyer $i$. For any given profile of types $\theta_{-i}$, let $\sigma_{-i}(\theta_{-i})=(\sigma_{1}(\theta_{1}),...,\sigma_{i-1}(\theta_{i-1}),\sigma_{i+1}(\theta_{i+1}),...,\sigma_{n}(\theta_{n}))$
denote the vector of experiments excluding bidder $i$. Let $M$ denote
the set of all two-stage mechanisms and let $m=(p,x,\sigma,\tau)$
denote a generic element of it.

\paragraph*{Payoffs:}

For any vector of realized types other than bidder $i$, let $\sigma_{-i}[\theta_{-i}]$
denote the joint density function $\sigma_{1}(\theta_{1})\times...\times\sigma_{i-1}(\theta_{i-1})\times\sigma_{i+1}(\theta_{i+1})\times...\times\sigma_{n}(\theta_{n})$
and let its corresponding c.d.f. be denoted by $F_{\sigma_{-i}[\theta_{-i}]}$.
Given a two-stage mechanism $m=(p,x,\sigma,\tau)$ in $M$, let 
\begin{align*}
u_{i}^{m}(t'_{i},t_{i}|f_{i},\sigma_{-i}) & =E_{\theta_{-i}}[E_{\sigma_{-i}[\theta_{-i}]}[t_{i}p_{i}(t'_{i},t_{-i};f_{i},\sigma_{-i}(\theta_{-i}))-x_{i}(t'_{i},t_{-i};f_{i},\sigma_{-i}(\theta_{-i}))]]\\
 & =\int_{\Theta_{-i}}\int_{T_{-i}}[t_{i}p_{i}(t'_{i},t_{-i};f_{i},\sigma_{-i}(\theta_{-i}))\\
 & \;\;\;-x_{i}(t'_{i},t_{-i};f_{i},\sigma_{-i}(\theta_{-i}))]dF_{\sigma_{-i}[\theta_{-i}]}(t_{-i})\;d(F_{r}\times F_{s})^{n-1}(\theta_{-i})
\end{align*}
be the stage-$2$ expected payoff bidder $i$ receives when she submits
a bid $t'_{i}\in[a,b]$ while she observes a value $t_{i}\in[a,b]$
given that she registered the experiment $f_{i}$ while all other
bidders truthfully registered an experiment from their offered menus.

Whenever bidder $i$ truthfully submits her bid, that is when $t'_{i}=t_{i}$,
we simply write $u_{i}^{m}(t{}_{i}|f_{i},\sigma_{-i})$. Let $u_{i}^{m}(f_{i}|\sigma_{-i})=E_{f_{i}}[u_{i}^{m}(t{}_{i}|f_{i},\sigma_{-i})]$
be the interim expected payoff bidder $i$ obtains when she bids truthfully
according to the distribution $f_{i}$. Given $m=(p,x;\sigma,\tau)\in M$,
let $v_{i}^{m}(\theta'_{i}|\theta_{i})=u_{i}^{m}(\sigma_{i}(\theta'_{i})|\sigma_{-i})-c_{\theta_{i}}(\sigma_{i}(\theta'_{i}))-\tau_{i}(\theta'_{i})$
denote the net expected payoff of bidder $i$ when other bidders truthfully
follow the distributions $\sigma_{-i}(\theta_{-i})$ in their offered
menus depending on their type realizations $\theta_{-i}=(r_{-i},s_{-i})$
while bidder $i$ with true type $\theta_{i}\in\Theta_{i}$ reveals
$\theta'_{i}=(r'_{i},s'_{i})$ to register the experiment $\sigma_{i}(\theta'_{i})\in D(\theta'_{i2})$
with the associated information processing cost $c_{\theta_{i}}(\sigma_{i}(\theta'_{i}))$
and transfer fee $\tau(\theta'_{i})$. Whenever bidder $i$ truthfully
submits her type $\theta_{i}=(r_{i},s_{i})$, we write $v_{i}^{m}(\theta_{i})$.
Let 
\[
U_{i}(m)=E_{\theta_{i}}[v_{i}^{m}(\theta_{i})]=\int_{\Theta_{i}}v_{i}^{m}(\theta_{i})d(F_{r}\times F_{s})(\theta_{i})
\]
 denote the ex-ante net expected payoff bidder $i$ receives when
she plans to truthfully reveal her type at stage-$1$ and reveal her
true value at stage-$2$ while all other bidders also truthfully submit
their types and bids given the profile of assignment rules $\sigma_{-i}$.

\paragraph*{Constraints:}

The seller does not observe the true type $\theta_{i}=(r_{i},s_{i})$
of a bidder nor the realized value $t_{i}$. The seller, therefore,
must provide the right incentives to each buyer to participate in
the mechanism, submit truthfully their realized types and values,
and conduct their intended experiments.

Given a two-stage auction mechanism $m=(p,x,\sigma,\tau)$, the constraints
listed below constitute the two-stage individual rationality and incentive
compatibility constraints for each buyer $i$ for each $\theta_{i}\in\Theta_{i}$:
\begin{enumerate}
\item $u_{i}^{m}(t{}_{i}|\sigma_{i}(\theta_{i}),\sigma_{-i})\geq0$ for
each $t_{i}\in[a,b]$ in the support of $\sigma_{i}(\theta_{i})$
\emph{(stage-$2$ IR);}
\item $u_{i}^{m}(t{}_{i}|\sigma_{i}(\theta_{i}),\sigma_{-i})\geq u_{i}^{m}(t'{}_{i},t_{i}|\sigma_{i}(\theta_{i}),\sigma_{-i})$
for each $t_{i},t'_{i}\in[a,b]$ in the support of $\sigma_{i}(\theta_{i})$
\emph{(stage-$2$ IC);}
\item $v_{i}^{m}(\theta_{i})\geq0$ \emph{(stage-$1$ IR);}
\item $v_{i}^{m}(\theta_{i})\geq v_{i}^{m}(\theta'_{i}|\theta_{i})$ for
each $\theta'_{i}\in\Theta_{i}$ \emph{(stage-$1$ IC).}
\end{enumerate}
The first two constraints are the usual individual rationality and
incentive compatibility constraints at the allocation stage. The third
constraint is the pre-allocation individual rationality constraint
implying that participation in the mechanism is voluntary. The fourth
constraint is the pre-allocation incentive compatibility constraint
incentivizing the bidders to tell the truth about their types (i.e.,
cost-scale and mean-value parameters).

We say that a two-stage mechanism $m=(p,x,\sigma,\tau)$ is feasible
if it satisfies the above four constraints. Let $\bar{M}$ denote
the set of feasible two-stage mechanisms. There are two separate objectives
that we focus on: revenue maximization and welfare maximization.

\paragraph*{Revenue maximization:}

Given $m=(p,x,\sigma,\tau)\in M$, let 
\begin{align*}
R(m) & =E_{\theta}[E_{\sigma(\theta)}\sum_{i\in N}[x_{i}(t,\sigma(\theta))+\tau_{i}(\theta_{i})]]\\
 & =\int_{\Theta}[\int_{T}[\sum_{i\in N}[x_{i}(t,\sigma(\theta))+\tau_{i}(\theta_{i})]]\,dF_{\sigma(\theta)}(t)\,]\,d(F_{r}\times F_{s})^{n}(\theta)
\end{align*}
 denote the total expected payment received by the seller from the
bidders if mechanism $(p,x)$ is run at stage-$2$ and profiles of
contracts $(\sigma,\tau)$ are used at stage-$1$, while the bidders
truthfully reveal their types. Let $R(m^{*})=\max_{m\in\bar{M}}R(m)$
denote the maximum level of expected revenue the seller can generate.

\paragraph*{Welfare maximization:}

Let $W(m)=\left[\sum_{i\in N}U_{i}(m)\right]+R(m)$ denote the total
expected welfare that can be achieved when mechanism$m=(p,x,\sigma,\tau)\in M$
is used. Let $W(m^{*})=\max_{m\in\bar{M}}W(m)$ denote the maximum
level of expected welfare that can be generated .

If a two-stage mechanism $m=(p,x,\sigma,\tau)$ maximizes the total
surplus $\sum_{i\in N}t_{i}p_{i}(t,\sigma(\theta))$ for each profile
of types $\theta\in\Theta$ and for each vector of values $t$ within
the support of $\sigma(\theta)$, we call it \emph{ex-post efficient};
when the mechanism $(p,x,\sigma,\tau)$ maximizes the expected (total)
net surplus 
\[
E_{\theta}[\sum_{i\in N}t_{i}p_{i}(t,\sigma(\theta))]-E_{\theta}[\sum_{i\in N}c_{\theta_{i}}(\sigma_{i}(\theta_{i}))],
\]
 we call it \emph{ex-ante efficient}. In other words, a mechanism
$m$ is ex-ante efficient if it maximizes the total expected welfare
(see also below).

\section{Analysis}

In this section, we study the optimal mechanisms that either maximize
the expected revenue or the expected welfare. We start our analysis
by first simplifying the objective functions within the class of feasible
two-stage mechanisms. Suppose that $m=(p,x,\sigma,\tau)$ is a given
feasible mechanism in $\bar{M}$. By definition, for each $\theta_{i}=(r_{i},s_{i})$
we have $v_{i}^{m}(\theta_{i})=u_{i}^{m}(\sigma_{i}(\theta_{i})|\sigma_{-i})-c_{\theta_{i}}(\sigma_{i}(\theta_{i}))-\tau_{i}(\theta_{i})$
where
\begin{align*}
u_{i}^{m}(\sigma_{i}(\theta_{i})|\sigma_{-i}) & =E_{\sigma_{i}(\theta_{i})}[E_{\theta_{-i}}[E_{\sigma_{-i}[\theta_{-i}]}[t_{i}p_{i}(t'_{i},t_{-i};\sigma_{i}(\theta_{i}),\sigma_{-i}(\theta_{-i}))\\
 & \;\;\;\;\;\;\;\;\;\;\;\;\;\;\;\;\;\;\;\;\;\;\;\;-x_{i}(t'_{i},t_{-i};\sigma_{i}(\theta_{i}),\sigma_{-i}(\theta_{-i}))]]].
\end{align*}
 Taking the expectation over $\Theta_{i}$, we obtain 
\[
E_{\theta_{i}}[v_{i}^{m}(\theta_{i})]=E_{\theta_{i}}[u_{i}^{m}(\sigma_{i}(\theta_{i})|\sigma_{-i})]-E_{\theta_{i}}[c_{\theta_{i}}(\sigma_{i}(\theta_{i}))]-E_{\theta_{i}}[\tau_{i}(\theta_{i})].
\]
Summing over bidders, we derive that 
\begin{align*}
E_{\theta}[\sum_{i\in N}v_{i}^{m}(\theta_{i})] & =E_{\theta}[\sum_{i\in N}t_{i}p_{i}(t,\sigma(\theta))-x_{i}(t,\sigma(\theta))]-E_{\theta}[\sum_{i\in N}c_{\theta_{i}}(\sigma_{i}(\theta_{i}))]-E_{\theta}[\sum_{i\in N}\tau_{i}(\theta_{i})].
\end{align*}
Rearranging the terms, we first obtain that the expected revenue of
the seller can be given as the difference between the expected net
surplus (i.e., surplus net of cost of information) and the expected
information rent: 
\begin{align*}
R(m) & =E_{\theta}[\sum_{i\in N}(x_{i}(t,\sigma(\theta))+\tau_{i}(\theta_{i}))]\\
 & =\underset{\text{expected net surplus}}{\underbrace{E_{\theta}[\sum_{i\in N}t_{i}p_{i}(t,\sigma(\theta))]-E_{\theta}[\sum_{i\in N}c_{\theta_{i}}(\sigma_{i}(\theta_{i}))]}}-\underset{\text{expected information rent}}{\underbrace{E_{\theta}[\sum_{i\in N}v_{i}^{m}(\theta_{i})]}.}
\end{align*}
We can further arrange the above equation to express the regulator's
welfare function as:
\begin{align*}
W(m) & =R(m)+E_{\theta}[\sum_{i\in N}v_{i}^{m}(\theta_{i})]\\
 & =\underset{\text{expected welfare}}{\underbrace{E_{\theta}[\sum_{i\in N}t_{i}p_{i}(t,\sigma(\theta))]-E_{\theta}[\sum_{i\in N}c_{\theta_{i}}(\sigma_{i}(\theta_{i}))]}}.
\end{align*}
The two expressions above clearly manifest the tension between revenue
and welfare maximization. As such, the seller's optimal mechanism
may distort the level of welfare that can be maximally obtained, and
similarly the efficient mechanism may cause the seller to leave additional
rent to the bidders lowering the revenue. Naturally we come across
similar efficiency-optimality conflict among many incentive compatible
mechanisms. For instance, in the setting of \citet{Myerson81} where
the distributions of values are fixed, the VCG mechanism $(p^{e},x^{e})$
delivers an efficient outcome, whereas the Myerson mechanism $(p^{r},x^{r})$
yields a revenue maximizing outcome. Since in our setting experiments
are endogenously determined, a conflict between revenue maximization
and welfare maximization, however, do not occur. In other words, the
seller's revenue maximizing mechanism will also be delivering the
maximum possible welfare. This means that the seller's optimal mechanism
would not only be maximizing the expected net surplus, but also would
be minimizing the expected information rent among the feasible two-stage
mechanisms.

\paragraph*{Efficient stage-$2$ rule:}

In the following, we will focus on two-stage mechanisms where at stage-$2$
the VCG mechanism, denoted by $(p^{e},x^{e})$, is used. The VCG mechanism
awards the good to the highest reported value $t_{i}$ (with a tie-break
rule) and the winner pays the highest reported value by the other
bidders. More formally, let $w(t)=\{i\in N:t_{i}\geq\max_{j\in N}t_{j}\}$
be the set of bidders who bid the highest value. Let 
\[
p_{i}^{e}(t)=\begin{cases}
1/|w(t)| & \text{if}\,i\in w(t),\\
0 & \text{if}\,i\notin w(t),
\end{cases}
\]
be the winning probability and let
\[
x_{i}^{e}(t)=\begin{cases}
\max_{j\neq i}t_{j} & \text{if}\,i\,\text{wins},\\
0 & \text{otherwise},
\end{cases}
\]
be the payment for all $i\in N$ and $t\in T$. That is, a bidder
wins only when she is among the highest bidders, and pays an amount
only when she wins and the payment is equal to the highest bid among
the opponents.

It is well-known that for any fixed profile of distributions $f=(f_{1},...,f_{n})$,
the VCG mechanism $(p^{e},x^{e})$ maximizes the total surplus among
the incentive compatible (IC) and individually rational (IR) allocation-payment
rules $(p,x)$; that is, for any stage-$2$ rule $(p,x)$ satisfying
the IC and IR constraints we have $\sum_{i\in N}t_{i}p_{i}^{e}(t,f)\geq\sum_{i\in N}t_{i}p_{i}(t,f).$
As such, for any given stage-$1$ contract $\sigma$, we have $E_{\theta}[\sum_{i\in N}t_{i}p_{i}^{e}(t,\sigma(\theta))]\geq E_{\theta}[\sum_{i\in N}t_{i}p_{i}(t,\sigma(\theta))]$
implying that 
\[
E_{\theta}[\sum_{i\in N}t_{i}p_{i}^{e}(t,\sigma(\theta))-c_{\theta_{i}}(\sigma_{i}(\theta_{i}))]\geq E_{\theta}[\sum_{i\in N}t_{i}p_{i}(t,\sigma(\theta))-c_{\theta_{i}}(\sigma_{i}(\theta_{i}))].
\]
But this means that a two-stage mechanism which uses the efficient
rule $(p^{e},x^{e})$ at stage-$2$ pointwise dominates (in terms
of the expected net surplus) a two-stage mechanism which uses another
rule $(p,x)$ at stage-$2$. In fact, the following result shows that
for any given feasible two-stage mechanism $m=(p,x,\sigma,\tau)\in\bar{M}$,
we can replace the stage-$2$ rule $(p,x)$ with the efficient rule
$(p^{e},x^{e})$ while we use the same stage-$1$ assignments $\sigma'=\sigma$
but adjust the fees $\tau$ to some $\tau'$ such that the new mechanism
is feasible and it weakly improves the revenue; that is, $(p^{e},x^{e},\sigma',\tau')\in\bar{M}$
and $R(p^{e},x^{e},\sigma',\tau')\geq R(p,x,\sigma,\tau)$.

Let $\bar{M}^{e}$ denote the subset of the set of feasible mechanisms
$\bar{M}$ such that the stage-$2$ rule is fixed as the efficient
second-price auction (i.e., the VCG mechanism).
\begin{lem}
\label{lem:vcg}For any $m=(p,x,\sigma,\tau)\in\bar{M}$, there exists
$m'=(p',x',\sigma',\tau')\in\bar{M}^{e}$ such that $R(m')\geq R(m)$.
\end{lem}
\begin{proof}
Let $m=(p,x,\sigma,\tau)\in\bar{M}$ be a feasible mechanism. Let
$m'=(p',x',\sigma',\tau')\in M$ be an alternative mechanism such
that $(p',x')=(p^{e},x^{e})$ and $\sigma'=\sigma$, and for any bidder
$i$, $\tau'_{i}(\theta_{i})=\tau_{i}(\theta_{i})+u_{i}^{m'}(\sigma_{i}(\theta_{i})|\sigma_{-i})-u_{i}^{m}(\sigma_{i}(\theta_{i})|\sigma_{-i})$.
We want to show that (i) $(p',x',\sigma',\tau')\in\bar{M}$ and (ii)
$R(p',x',\sigma',\tau')\geq R(p,x,\sigma,\tau)$. As discussed before,
constraints $1$ and $2$ hold when the VCG mechanism $(p^{e},x^{e})$
is used at stage-$2$. Let $v_{i}^{m'}(\theta_{i})$ be the net expected
payoff of bidder $i$ when she truthfully reveals her type $\theta_{i}=(r_{i},s_{i})$
under the mechanism $m'$. By definition, we have
\begin{align*}
v_{i}^{m'}(\theta_{i}) & =u_{i}^{m'}(\sigma'_{i}(\theta_{i})|\sigma'_{-i})-c_{\theta_{i}}(\sigma'_{i}(\theta_{i}))-\tau'_{i}(\theta_{i})\\
 & =u_{i}^{m}(\sigma_{i}(\theta_{i})|\sigma_{-i})-c_{\theta_{i}}(\sigma_{i}(\theta_{i}))-\tau{}_{i}(\theta_{i})\\
 & =v_{i}^{m}(\theta_{i})\geq0
\end{align*}
 showing that constraint $3$ holds under $m'$. Now let $v_{i}^{m'}(\theta'_{i}|\theta_{i})$
denote the net expected payoff of bidder $i$ when she reveals her
type as $\theta'_{i}=(r'_{i},s'_{i})$ instead of telling her true
type $\theta_{i}=(r_{i},s_{i})$. By definition, we have
\begin{align*}
v_{i}^{m'}(\theta'_{i}|\theta_{i}) & =u_{i}^{m'}(\sigma'_{i}(\theta'_{i})|\sigma'_{-i})-c_{\theta_{i}}(\sigma'_{i}(\theta'_{i}))-\tau'_{i}(\theta'_{i})\\
 & =u_{i}^{m}(\sigma_{i}(\theta'_{i})|\sigma{}_{-i})-c_{\theta_{i}}(\sigma_{i}(\theta'_{i}))-\tau_{i}(\theta'_{i})\\
 & =v_{i}^{m}(\theta'_{i}|\theta_{i})\leq v_{i}^{m}(\theta_{i})=v_{i}^{m'}(\theta_{i})
\end{align*}

showing that constraint $4$ holds under $m'$. Thus, $m'\in\bar{M}$.
Moreover we have, 
\begin{align*}
R(m) & =E_{\theta}[\sum_{i\in N}t_{i}p_{i}(t,\sigma(\theta))]-E_{\theta}[\sum_{i\in N}c_{\theta_{i}}(\sigma_{i}(\theta_{i}))]-E_{\theta}[\sum_{i\in N}v_{i}^{m}(\theta_{i})]\\
 & =E_{\theta}[\sum_{i\in N}t_{i}p_{i}(t,\sigma'(\theta))]-E_{\theta}[\sum_{i\in N}c_{\theta_{i}}(\sigma'_{i}(\theta_{i}))]-E_{\theta}[\sum_{i\in N}v_{i}^{m'}(\theta_{i})]\\
 & \leq E_{\theta}[\sum_{i\in N}t_{i}p_{i}^{e}(t,\sigma'(\theta))]-E_{\theta}[\sum_{i\in N}c_{\theta_{i}}(\sigma'_{i}(\theta_{i}))]-E_{\theta}[\sum_{i\in N}v_{i}^{m'}(\theta_{i})]\\
 & =R(m').
\end{align*}
where the inequality follows from the fact that the VCG rule $(p^{e},x^{e})$
maximizes the surplus for any fixed profile of distributions.
\end{proof}
Lemma \ref{lem:vcg} shows that in searching for an optimal mechanism
maximizing the expected revenue, one can focus on the set of mechanisms
which employ the efficient rule at stage-$2$. In fact, the proof
of this result shows that from the given mechanism $m$ an alternative
mechanism $m'$ can be constructed such that each bidder's information
rent will remain the same; that is, $v_{i}^{m'}(\theta_{i})=v_{i}^{m}(\theta_{i})$
for all $i\in N$ and $\theta_{i}\in\Theta_{i}$. As such, we also
have
\begin{align*}
W(m) & =R(m)+E_{\theta}[\sum_{i\in N}v_{i}^{m}(\theta_{i})]\\
 & \leq R(m')+E_{\theta}[\sum_{i\in N}v_{i}^{m'}(\theta_{i})]=W(m')
\end{align*}
 showing that to maximize the expected welfare, one can focus on the
efficient mechanisms in $\bar{M}^{e}$. Thus we have the following
corollary.
\begin{cor}
\label{cor:efficient}For any $m=(p,x,\sigma,\tau)\in\bar{M}$, there
exists $m'=(p',x',\sigma',\tau')\in\bar{M^{e}}$ such that $W(m')\geq W(m)$.
\end{cor}
Corollary \ref{cor:efficient} impies that shows any feasible ex-ante
efficient two-stage mechanism must use the VCG rule at stage-$1$,
and, thus, by Lemma \ref{lem:vcg} any revenue maximizing two-stage
mechanism must be ex-ante efficient. There are related studies in
the literature which show, similar to the above corollary, that the
VCG rule maximizes the welfare whenever bidder participation or preference
determination is costly. For instance, in an IPV auction setting where
sending messages to the seller (i.e., participation in the auction)
is costly, \citet{Stegeman96} shows that the second-price auction
has an equilibrium that is ex-ante efficient. In another study with
an IPV auction setting where bidders can acquire (covert) information
before participating in the auction, \citet{BV02} show that the VCG
rule guarantees both ex-ante and ex-post efficiency. When we have
an efficient mechanism $(p,x,\sigma,\tau)\in\bar{M}^{e}$ and we want
to focus on the assignment rule $\sigma$, we also write $W^{e}(\sigma)$
to denote the expected welfare under the mechanism $(p,x,\sigma,\tau)$
before fees are applied.

\paragraph*{Potential game:}

Given that the optimal stage-$2$ rule should be the efficient rule
$(p^{e},x^{e})$, we are now interested in understanding the optimal
stage-$1$ assignments $\sigma$ and schedules $\tau$ which determine
an experiment $\sigma_{i}(\theta_{i})$ and a fee $\tau{}_{i}(\theta_{i})$
to each declared type $\theta_{i}\in\Theta_{i}$ by each bidder $i$.
Notice that stage-$1$ fees do not depend on the experiments registered
by the bidders. As such, the only remaining constraint for bidder
$i$'s choice of experiment $\sigma_{i}(\theta_{i})$ is that its
center $\mu(\sigma_{i}(\theta_{i}))$ must meet the declared mean
value $z_{\theta_{i2}}$; that is, $\sigma_{i}(\theta_{i})\in D(\theta_{i2})$.
Beyond this restriction the bidders are (and should be) free to choose
the best experiment they can obtain based on their types. The reason
for this fact is that we have a Bayesian exact potential game where
each bidder's switch of an experiment leads to the same amount of
change in the expected welfare and the expected payoff of the bidder.
To see this, let $m,m'\in\bar{M}^{e}$ be two given efficient mechanisms
with assignment and fee rules $(\sigma,\tau),(\sigma',\tau')$, respectively
such that $\tau=\tau'$ and $\sigma_{-i}=\sigma'_{-i}$, but $\sigma_{i}\neq\sigma_{i}'$
for some bidder $i\in N$. Note that under the efficient rule $(p^{e},x^{e})$,
for any given profile of bids $t\in T$, payoff of bidder $i$ is
\[
u_{i}^{e}(t)=t_{i}p_{i}^{e}(t)+\sum_{j\neq i}t_{j}p_{j}^{e}-\max_{j\neq i}t_{j}=\sum_{k\in N}t_{k}p_{k}^{e}-\max_{j\neq i}t_{j}.
\]
Therefore, the interim payoff of bidder $i$ when she bids $t_{i}$
is
\begin{align*}
u_{i}^{e}(t_{i}|\sigma_{i}(\theta_{i}),\sigma_{-i}) & =E_{\theta_{-i}}[E_{\sigma_{-i}[\theta_{-i}]}[t_{i}p_{i}^{e}(t|\sigma_{i}(\theta_{i}),\sigma_{-i})+\sum_{j\neq i}t_{j}p_{j}^{e}(t|\sigma_{i}(\theta_{i}),\sigma_{-i})-\max_{j\neq i}t_{j}]]\\
 & =E_{\theta_{-i}}[E_{\sigma_{-i}[\theta_{-i}]}[\sum_{k\in N}t_{k}p_{k}^{e}(t|\sigma_{i}(\theta_{i}),\sigma_{-i})-\max_{j\neq i}t_{j}]],
\end{align*}
given that all other bidders truthfully bid according to their assigned
experiments in $\sigma_{-i}(\theta_{-i})$ for each realization of
profile of types $\theta_{-i}$. Therefore, the interim payoff of
bidder $i$ when she follows her experiment $\sigma_{i}(\theta_{i})$
is equal to
\begin{align*}
u_{i}^{e}(\sigma_{i}(\theta_{i})|\sigma_{-i}) & =E_{\sigma_{i}(\theta_{i})}[E_{\theta_{-i}}[E_{\sigma_{-i}[\theta_{-i}]}[\sum_{k\in N}t_{k}p_{k}^{e}(t|\sigma_{i}(\theta_{i}),\sigma_{-i}(\theta_{-i}))]]]\\
 & \;\;\;\;\;\;\;\;\;\;\;\;-E_{\theta_{-i}}[E_{\sigma_{-i}[\theta_{-i}]}[\max_{j\neq i}t_{j}]].
\end{align*}
 But then when there is a change of assignment for bidder $i$ from
$\sigma_{i}$ to $\sigma'_{i}$, for each type realization $\theta_{i}\in\Theta_{i}$
we have the payoff difference as
\begin{align*}
u_{i}^{e}(\sigma'_{i}(\theta_{i})|\sigma_{-i})-u_{i}^{e}(\sigma_{i}(\theta_{i})|\sigma_{-i}) & =E_{\sigma'_{i}(\theta_{i})}[E_{\theta_{-i}}[E_{\sigma_{-i}[\theta_{-i}]}[\sum_{k\in N}t_{k}p_{k}^{e}(t|\sigma'_{i}(\theta_{i}),\sigma_{-i}(\theta_{-i}))]]]\\
 & \;\;\;-E_{\sigma_{i}(\theta_{i})}[E_{\theta_{-i}}[E_{\sigma_{-i}[\theta_{-i}]}[\sum_{k\in N}t_{k}p_{k}^{e}(t|\sigma_{i}(\theta_{i}),\sigma_{-i}(\theta_{-i}))]]].
\end{align*}
Taking an expectation over $\Theta_{i}$, we derive that the expected
payoff difference for bidder $i$ by switching from $\sigma_{i}$
to $\sigma'_{i}$ is equal to
\begin{align}
E_{\theta_{i}}[u_{i}^{e}(\sigma'_{i}(\theta_{i})|\sigma_{-i})]-E_{\theta_{i}}[u_{i}^{e}(\sigma_{i}(\theta_{i})|\sigma_{-i})] & =E_{\theta}[\sum_{k\in N}t_{k}p_{k}^{e}(t,\sigma'(\theta))]-E_{\theta}[\sum_{k\in N}t_{k}p_{k}^{e}(t,\sigma(\theta))].\nonumber \\
 & \;\;\;\;\;\;\label{eq:potential}
\end{align}
Let $\Pi_{i}^{e}(\sigma_{i},\sigma_{-i})=E_{\theta_{i}}[u_{i}^{e}(\sigma_{i}(\theta_{i})|\sigma_{-i})]-E_{\theta_{i}}[c_{\theta_{i}}(\sigma_{i}(\theta_{i}))]$
denote the expected net interim payoff of bidder $i$ before the fees
are applied. Whenever the fees are set in such a way that the stage-$1$
IC and IR constraints hold, then the bidder will be maximizing the
objective function $\Pi^{e}(\sigma_{i},\sigma_{-i})$ by choosing
an experiment $\sigma_{i}(\theta_{i})$ for each realization of her
type $\theta_{i}\in\Theta_{i}$. Notice that if we subtract $E_{\theta_{i}}[c_{\theta_{i}}(\sigma'_{i}(\theta_{i}))]-E_{\theta_{i}}[c_{\theta_{i}}(\sigma_{i}(\theta_{i}))]$
from the left hand side in the equation (\ref{eq:potential}), then
we obtain $\Pi_{i}^{e}(\sigma'_{i},\sigma_{-i})-\Pi_{i}^{e}(\sigma_{i},\sigma_{-i})$.

By definition, we have 
\[
W^{e}(\sigma_{i},\sigma_{-i})=E_{\theta}[\sum_{k\in N}t_{k}p_{k}^{e}(t,\sigma(\theta)]-E_{\theta}[\sum_{k\in N}c_{\theta_{k}}(\sigma_{k}(\theta_{k}))]
\]
 as the expected welfare. If we subtract $E_{\theta}[\sum_{k\in N}c_{\theta_{k}}(\sigma'_{k}(\theta_{k}))]-E_{\theta}[\sum_{k\in N}c_{\theta_{k}}(\sigma_{k}(\theta_{k}))]$
from the right hand side of the equation (\ref{eq:potential}), we
obtain $W^{e}(\sigma'_{i},\sigma_{-i})-W^{e}(\sigma_{i},\sigma_{-i})$.
But also note that since $\sigma_{-i}=\sigma'_{-i}$, when switching
from $\sigma_{i}$ to $\sigma'_{i}$ the expected cost difference
for bidder $i$ and the total expected cost difference are the same,
and so we have
\[
E_{\theta_{i}}[c_{\theta_{i}}(\sigma'_{i}(\theta_{i}))]-E_{\theta_{i}}[c_{\theta_{i}}(\sigma_{i}(\theta_{i}))]=E_{\theta}[\sum_{k\in N}c_{\theta_{k}}(\sigma'_{k}(\theta_{k}))]-E_{\theta}[\sum_{k\in N}c_{\theta_{k}}(\sigma_{k}(\theta_{k}))].
\]

This means, however, that the following equality holds
\[
\Pi_{i}^{e}(\sigma'_{i},\sigma_{-i})-\Pi_{i}^{e}(\sigma_{i},\sigma_{-i})=W^{e}(\sigma'_{i},\sigma_{-i})-W^{e}(\sigma_{i},\sigma_{-i})
\]
 which shows that we have a Bayesian exact potential game with the
potential function $W^{e}$; that is, the marginal change in the expected
welfare is the same as the marginal change in the payoff of bidder
$i$ when she switches from $\sigma_{i}$ to $\sigma'_{i}$. As such,
the game bidders are playing at stage-$1$ is a Bayesian exact potential
game, and so any maximizer $\sigma^{*}$ of the expected welfare function
$W^{e}(\sigma)$ must be a Bayesian Nash equilibrium of the game where
each bidder freely chooses an experiment $f_{i}\in D(s_{i})$ whenever
her type is some $\theta_{i}=(r_{i},s_{i})$.\footnote{In a related study, \citet{BP07} consider an IPV auction setting
and show that when the seller is in full control of the information
precision of each bidder, then the optimal information structures
are partitional and the partitions are asymmetric among the bidders.}

\paragraph*{Optimal assignments: }

We now show that the potential game we described above has a solution
$\sigma^{*}$. For any given type $\theta_{i}$, the set $D(\theta_{i2})$
is a non-empty, convex, closed subset of $\Delta([a,b])$. Recall
that the set $\Delta([a,b])$ is weakly compact, and so $D(\theta_{i2})$
is also compact. A pure (measurable) strategy for bidder $i$ is a
map $\sigma_{i}:\Theta_{i}\to\Delta([a,b])$ such that $\sigma_{i}(\theta_{i})\in D(\theta_{i2})$
for almost all $\theta_{i}\in\Theta_{i}$ with respect to the measure
$F_{r}\times F_{s}$. Let $A_{i}$ be the space of such maps. We equip
$A_{i}$ with the pointwise weak convergence topology since $A_{i}\subset\prod_{\theta_{i}}\Delta([a,b])$,
and therefore it inherits the product of weak topologies. Moreover,
$A_{i}$ is closed and since $\prod_{\theta_{i}}\Delta([a,b])$ is
compact by Tyschonoff's theorem, $A_{i}$ is compact. Now let $A=A_{1}\times...\times A_{n}$
be the joint strategy space with the product of the product-weak topologies.
$A$ is compact since it is the product of compact spaces.

Given $\sigma=(\sigma_{1},...,\sigma_{n})$, we have the Bayesian
exact potential as
\[
W^{e}(\sigma)=E_{\theta}[\sum_{k\in N}t_{k}p_{k}^{e}(t,\sigma(\theta))]-E_{\theta}[\sum_{k\in N}c_{\theta_{k}}(\sigma_{k}(\theta_{k}))].
\]
 We want to show that $W^{e}(\sigma)$ is upper semi-continuous (u.s.c).
Clearly, the first term above is continuous since it is a linear function
of $\sigma$. For each $\theta_{i}$, the cost function $c_{\theta_{i}}(.)$
is l.s.c. on $D(\theta_{i2})$. Thus the function $c_{\theta_{i}}(\sigma_{i}(\theta_{i}))$
is l.s.c. pointwise. Taking an expectation over $\Theta_{i}$ yields
that $E_{\theta_{i}}[c_{\theta_{i}}(\sigma_{i}(\theta_{i}))]$ is
l.s.c., and finally summing over the bidders gives us the fact that
$E_{\theta}[\sum_{k\in N}c_{\theta_{k}}(\sigma_{k}(\theta_{k}))]$
is a l.s.c. functional of $f_{\sigma}$. Thus, we derive that $W^{e}(\sigma)$
is an u.s.c functional of $\sigma$. Therefore, by the Weierstrass
theorem there exists at least one maximizer $\sigma^{*}$ such that
$W^{e}(\sigma^{*})=\max_{\sigma\in A}W^{e}(\sigma)$.

Furthermore, we can argue that there exists a symmetric solution such
that $\sigma_{i}^{*}=\sigma_{j}^{*}$ for all $i,j\in N$. Now let
$\pi$ be a permutation of the set $N$ and define by $\sigma_{\pi}^{*}=(\sigma_{\pi(1)}^{*},...,\sigma_{\pi(n)}^{*})$
the permuted profile. Since every bidder's type realization is i.i.d,
both the first and second terms in the welfare functional are symmetric
in the bidders, we have $W^{e}(\sigma^{*})=W^{e}(\sigma_{\pi}^{*})$
for all $\pi$. Thus, every permutation $\sigma_{\pi}^{*}$ of the
maximizer $\sigma^{*}$ is also a maximizer. Let $\bar{\sigma}=\frac{1}{n!}\sum_{\pi}\sigma_{\pi}^{*}$
be a symmetrized strategy profile. Since $A$ is convex, we have $\bar{\sigma}\in A$.
Recall that the first term of the welfare functional is linear in
$\sigma$, while the second term is convex in $\sigma$ since $c_{\theta_{i}}(.)$
is a convex function for each $i$ and $\theta_{i}$. But then the
welfare functional $W^{e}(\sigma)$ is concave in $\sigma$. Therefore,
by Jensen's inequality we have 
\[
W^{e}(\bar{\sigma})=W^{e}(\frac{1}{n!}\sum_{\pi}\sigma_{\pi}^{*})\geq\frac{1}{n!}\sum_{\pi}W^{e}(\sigma_{\pi}^{*})=W^{e}(\sigma^{*})
\]
 showing that the symmetrized strategy profile $\bar{\sigma}$ is
a maximizer of $W^{e}$. Is the solution unique? With our current
set of assumptions on the primitives, this is not necessarily the
case. However, if we assume that the cost function $c_{\theta_{i}}$
is strictly convex instead of only convex, then the potential function
becomes strictly concave. In that case, there exists a unique solution.
Moreover, due to the permutation argument given above, the unique
solution must be symmetric.

\paragraph*{Optimal fees: }

We now analyze how the optimal fees $\tau^{*}$ can be obtained. Recall
that we fixed $(p^{e},x^{e})$ as stage-$2$ rule and without loss
of generality the optimal assignment profile $\sigma^{*}$ at stage-$1$
can be taken symmetric such that $\sigma_{i}^{*}=\sigma_{j}^{*}$
for all bidders $i,j$. For each truly reported type $\theta_{i}=(r_{i},s_{i})$,
let $\Phi_{i}(r_{i},s_{i})=\max_{f_{i}\in D(s_{i})}\left(u_{i}^{e}(f_{i}|\sigma_{-i}^{*})-c_{r_{i},s_{i}}(f_{i})\right)$
denote the highest stage-$2$ utility obtained before the fees applied.
If instead the type is misreported as $\theta'_{i}=(r'_{i},s'_{i})$,
then her best (before fees are applied) utility is $\Phi_{i}(r'_{i},s'_{i}|r_{i},s_{i})=\max_{f_{i}\in D(s'_{i})}\left(u_{i}^{e}(f_{i}|\sigma_{-i}^{*})-c_{r_{i},s_{i}}(f_{i})\right)$.

Note that the misreported type $\theta'_{i}=(r'_{i},s'_{i})$ affects
pre-fee payoff only through $s'_{i}$; misreported $r'_{i}$ has no
effect. Fix some $s_{i}$ and let $\theta_{i}=(r_{i},s_{i})$ be the
true type and $(r'_{i},s_{i})$ be the reported type. Let $\tau_{i}:\Theta_{i}\to\mathbb{R}$
be a feasible fee scheme for each bidder $i$. As such, $\tau_{i}$
(together with stage $2$ rule $(p^{e},x^{e})$ and stage $1$ experiment
scheme $\sigma^{*}$) must satisfy stage-$1$ IR and IC constraints.
The bidder $i$'s net payoff under the fee schedule $\tau_{i}$ is
$\Phi_{i}(r_{i},s_{i})-\tau_{i}(r'_{i},s_{i})$. The stage-$2$ IC
condition in $r$ then implies that 
\[
\Phi_{i}(r_{i},s_{i})-\tau_{i}(r_{i},s_{i})\geq\Phi_{i}(r_{i},s_{i})-\tau_{i}(r'_{i},s_{i})
\]
 and so $\tau_{i}(r_{i},s_{i})\leq\tau_{i}(r'_{i},s_{i})$. But then
swapping the roles of $r_{i}$ and $r'_{i}$, we obtain that $\tau_{i}(r_{i},s_{i})=\tau_{i}(r'_{i},s_{i})$
for all $s_{i}$; that is, the optimal fee should be independent of
the reported parameter $r_{i}$. Therefore, let $\tau_{i}(s_{i}):[0,1]\to\mathbb{R}$
denote in short the feasible fee for each bidder $i$ and reported
$s_{i}$.

The stage-$1$ IR constraint requires that $\Phi_{i}(r_{i},s_{i})-\tau_{i}(s_{i})\geq0$
for each $r_{i}$ and $s_{i}$. Notice that $\Phi_{i}(\bar{r}_{i},s_{i})\leq\Phi_{i}(r_{i},s_{i})$
for all $s_{i}$ where $\bar{r}_{i}=1$. Thus, the fee schedule must
satisfy 
\begin{equation}
\tau_{i}(s_{i})\leq\Phi_{i}(\bar{r}_{i},s_{i})\label{eq:IR}
\end{equation}
 for each $s_{i}\in[0,1]$. Now consider the stage-$1$ IC constraint
in $s_{i}$. We have 
\[
\Phi_{i}(r_{i},s_{i})-\tau_{i}(s_{i})\geq\Phi_{i}(r{}_{i},s'_{i}|r_{i},s_{i})-\tau_{i}(s'_{i})
\]

for every $r_{i}$, $s_{i}$, and $s'_{i}$. Let $H_{i}(s_{i},s'_{i})=\inf_{r_{i}\in[0,1]}\left(\Phi_{i}(r_{i},s_{i})-\Phi_{i}(r{}_{i},s'_{i}|r_{i},s_{i})\right)$
denote the lowest pre-fee utility gap among the true types $(r_{i},s_{i})$
who report $(r_{i},s'_{i})$. Due to the above stage-$1$ IC constraint,
we must have
\begin{equation}
\tau_{i}(s_{i})\leq\tau_{i}(s'_{i})+H_{i}(s_{i},s'_{i})\label{eq:IC}
\end{equation}
for each $s_{i},s'_{i}\in[0,1]$. It is clear that the above two conditions
(equations \ref{eq:IR} and \ref{eq:IC}) characterize the set of
all feasible stage-$1$ fees (see, e.g., \citet{Rochet87}, Theorem
1).

Now we want to find out which feasible fee can maximize the expected
revenue for the seller. Let $\varOmega=\cup_{k\in\mathbb{N}}[0,1]^{k}$
denote the set of finite chains (e.g., $\omega=(\omega_{1},...,\omega_{k})$
for some $k\in\mathbb{N}$) generated from the set $[0,1]$. Let $\Omega(s)$
denote the set of finite chains which end with $s\in[0,1]$; that
is, if $\omega\in\Omega(s)$ with $\omega\in[0,1]^{k}$ for some $k\geq0$,
then $\omega_{k}=s$. Let $l(\omega)$ denote the length of a given
chain $\omega\in\Omega$; that is, $l(\omega)=k\in\mathbb{N}$ if
$\omega\in[0,1]^{k}$. Define the chain-closure fee for each bidder
$i$ as follows:
\[
\tau_{i}^{cc}(s_{i})=\inf_{\omega\in\Omega(s_{i})}\left(\Phi_{i}(\bar{r}_{i},\omega_{1})+\sum_{k=2}^{l(\omega)}H_{i}(\omega_{k},\omega_{k-1})\right)
\]
for any $s_{i}\in[0,1]$. We now show that the chain-closure fee is
feasible; that is, it satisfies stage-$1$ IR and IC constraints.
Since $H_{i}(\omega_{k},w_{k-1})\geq0$ for any $\omega\in\Omega$,
we have $\inf_{\omega\in\Omega}\left(\Phi_{i}(\bar{r}_{i},\omega_{1})+\sum_{k=2}^{l(\omega)}H_{i}(\omega_{k},\omega_{k-1})\right)\leq\Phi_{i}(\bar{r}_{i},\omega_{1})$,
and so $\tau^{cc}(s_{i})\leq\Phi_{i}(\bar{r}_{i},s_{i})$ for any
given $s_{i}\in[0,1]$ showing that stage-$1$ IR holds. Now fix some
$s_{i},s'_{i}\in[0,1]$. Since $\tau^{cc}(s'_{i})$ is the infimum,
for any $\epsilon>0$ we can find a chain $\omega\in\Omega(s'_{i})$
such that $\Phi_{i}(\bar{r}_{i},\omega_{1})+\sum_{k=2}^{l(\omega)}H_{i}(\omega_{k},\omega_{k-1})\leq\tau^{cc}(s'_{i})+\epsilon$.
But then if we add $H_{i}(s_{i},s'_{i})$ to both hand side of the
inequality we derive
\[
\Phi_{i}(\bar{r}_{i},\omega_{1})+\sum_{k=2}^{l(\omega)}H_{i}(\omega_{k},\omega_{k-1})+H_{i}(s_{i},s'_{i})\leq\tau^{cc}(s'_{i})+H_{i}(s_{i},s'_{i})+\epsilon.
\]
By definition, the left hand side of the above inequality is bigger
than $\tau^{cc}(s_{i})$, and so we have $\tau_{i}^{cc}(s_{i})\leq\tau_{i}^{cc}(s'_{i})+H_{i}(s_{i},s'_{i})+\epsilon$.
Since $\epsilon$ was arbitrary, we can let it go to $0$, and obtain
in the limit that $\tau_{i}^{cc}(s_{i})\leq\tau_{i}^{cc}(s'_{i})+H_{i}(s_{i},s'_{i})$
showing that stage-$1$ IC condition holds. Thus, the chain-closure
fee $\tau_{i}^{cc}$ is feasible.

Let $\omega\in\Omega(s_{i})$ be a given finite chain for some $s_{i}\in[0,1]$.
Since for any feasible fee $\tau_{i}$ stage-$1$ IR and IC conditions
hold, we have $\tau_{i}(\omega_{1})\leq\Phi_{i}(\bar{r}_{i},\omega_{1})$
and $\tau_{i}(\omega_{k})\leq\tau_{i}(\omega_{k-1})+H_{i}(\omega_{k},\omega_{k-1})$
for any $k=2,...,l(\omega)$. By iterating and summing up the terms,
we obtain that $\tau_{i}(s_{i})=\tau_{i}(\omega_{k})\leq\Phi_{i}(\bar{r}_{i},\omega_{1})+\sum_{k=2}^{l(\omega)}H_{i}(\omega_{k},\omega_{k-1})$
showing that $\tau_{i}(s_{i})$ is a lower-bound for the set $T(s_{i})=\{\Phi_{i}(\bar{r}_{i},\omega_{1})+\sum_{k=2}^{l(\omega)}H_{i}(\omega_{k},\omega_{k-1}):\omega\in\Omega(s_{i})\}$.
Since, by definition, $\tau_{i}^{cc}(s_{i})$ is the infimum of the
set $T(s_{i})$, we conclude that $\tau_{i}(s_{i})\leq\tau_{i}^{cc}(s_{i})$
and so the chain-closure fee yields the maximum amount of fee for
each $s_{i}$ given that stage-$2$ allocation rule and stage-$1$
experiment rule are fixed as $(p^{e},x^{e})$ and $\bar{\sigma}$,
respectively.

\paragraph*{Optimal two-stage mechanism: }

To sum up, the optimal two-stage mechanism $(p,x,\sigma,\tau)$ in
$\bar{M}$ is obtained by letting stage-$2$ rule to be efficient
rule $(p^{e},x^{e}),$ while letting stage-$1$ assignment rule $\sigma$
to be a symmetric maximizer $\bar{\sigma}$ of the potential function
$W^{e}(\sigma)$, and letting stage-$1$ fees $\tau$ to be determined
by the chain-closure fees $\tau^{cc}$.
\begin{thm}
\label{thm:opt}The two-stage mechanism $m^{*}=(p^{e},x^{e},\bar{\sigma},\tau^{cc})$
maximizes both the expected revenue $R(m)$ and the expected welfare
$W(m)$ among the feasible mechanisms in $\bar{M}$. 
\end{thm}
The proof of Theorem \ref{thm:opt} follows from above discussions
(including Lemma \ref{lem:vcg} and Corollary \ref{cor:efficient}),
therefore it is omitted. We see that $(p^{e},x^{e},\bar{\sigma},\tau^{cc})$
is the optimal two-stage mechanism where at stage-$2$ the efficient
outcome is achieved, while at stage-$1$ the seller lets the bidders
experiment as they wish, but charges them a fee according to the chain-closure
scheme and generates the maximum level of expected revenue, while
leaving a minimum level of expected information rent to the bidders.
From the welfare maximization perspective this outcome is optimal
since the efficient rule already ensures that the maximum level of
welfare is achieved.

\section{Extensions}

We now consider some possible extensions of the mechanism design problem
we have considered above. Firstly, we relax the assumption that the
the seller can verify for free the experiment used $f_{i}$ or its
center $\mu(f_{i})$ whenever the bidder $i$ reports $s_{i}$. We
instead assume that the seller pays a cost for auditing each bidder's
choice of experiment. We identify the effects of this change on the
equilibrium outcomes. In an extreme case, the seller cannot verify
any statistics about the experiment (no auditing); in that situation,
we observe that the optimal transfer is a flat fee. Secondly, we allow
for the possibility that the principal can also verify the cost parameter
$r\in[0,1]$ of each bidder with costly auditing. We identify the
effects of this change on the optimal fee schedule.

\paragraph*{Auditing bidder experiments:}

Suppose that the seller cannot verify for free the experiment undertaken
by a bidder. Instead, the seller must pay a cost for auditing the
experiments. Let $k_{e}:[0,1]\to\mathbb{R}_{+}$ denote the cost of
auditing such that when the seller incurs the cost $k_{e}(q)\geq0$,
then with probability $q\in[0,1]$ the seller can perfectly observe
the conducted experiment, and with probability $1-q$ he does not
observe the experiment at all. In case a misreport by bidder $i$
is detected, then bidder $i$ is punished by at most an amount of
$P_{i}\leq0$.

As before, the optimal two-stage mechanism with auditing will be used
$(p^{e},x^{e},\bar{\sigma},\tau^{cc})$ as part of the revenue maximization
process. Moreover, the seller needs to determine an auditing rule
where for each $s_{i}\in[0,1]$ reported by bidder $i$, the seller's
success probability of catching a non-compliance will be $q(s_{i})\in[0,1]$.
For each reported type $\theta_{i}=(r_{i},s_{i})$, let 
\[
\Psi_{i}(r_{i},s_{i})=\max_{f_{i}\in D\setminus D(s_{i})}\left(u_{i}^{e}(f_{i}|\sigma_{-i}^{*})-c_{r_{i},s_{i}}(f_{i})\right)
\]
 denote the best non-compliant pre-fee utility of bidder $i$. If
bidder $i$ reports $\theta_{i}=(r_{i},s_{i})$ truthfully and complies,
she obtains a payoff of $\Phi_{i}(r_{i},s_{i})-\tau_{i}^{cc}(s_{i})$.
If, on the other hand, she deviates the best she can expect to receive
will be $(1-q(s_{i}))\,\Psi_{i}(r_{i},s_{i})\,+\,q(s_{i})\,P_{i}-\tau_{i}^{cc}(s_{i})$.
Therefore, truth-telling is enforced if $\Phi_{i}(r_{i},s_{i})\geq(1-q(s_{i}))\,\Psi_{i}(r_{i},s_{i})\,+\,q(s_{i})\,P_{i}$
for all $r_{i}$.

Rearranging the above terms, we have the lower-bound condition on
the deterrence probability $q(s_{i})\geq\frac{\Psi_{i}(r_{i},s_{i})-\Phi_{i}(r_{i},s_{i})}{\Psi_{i}(r_{i},s_{i})-P_{i}}\in[0,1]$.
Thus, the minimal auditing probability when $s_{i}$ is reported (regardless
of $r_{i}$) is 
\[
q_{P_{i}}^{*}(s_{i})=\sup_{r_{i}\in[0,1]}\left[\frac{\Psi_{i}(r_{i},s_{i})-\Phi_{i}(r_{i},s_{i})}{\Psi_{i}(r_{i},s_{i})-P_{i}}\right]^{+}
\]
where $[.]^{+}$refers to the $\max\{0,.\}$ operator. Note that the
higher $|P_{i}|$ is, the lower $q_{P_{i}}^{*}(s_{i})$ is for each
$s_{i}$, and so the lower auditing cost $k_{e}(q_{P_{i}}^{*}(s_{i}))$
will be. Suppose that the auditing costs are low enough that it is
optimal for the seller to audit. In that case, if $P=(P_{1},...,P_{n})$
is the profile of maximum punishments that can be imposed on bidders,
then in equilibrium the maximum expected revenue under auditing experiments
will be $E_{\theta}[\sum_{i\in N}(x_{i}^{e}(t,\sigma^{*}(\theta))+\tau_{i}^{cc}(\theta_{i2}))]-E_{\theta}[\sum_{i\in N}k_{e}(q_{P_{i}}^{*}(\theta_{i2}))]$.

\paragraph*{Auditing bidder costs:}

Suppose that the seller can verify for free the experiments. Moreover,
suppose that the seller can audit bidder costs. Specifically, let
$k_{r}:[0,1]\to\mathbb{R}_{+}$ denote the cost of auditing such that
when the seller incurs the cost $k_{r}(q)\geq0$, then with probability
$q\in[0,1]$ the seller can perfectly observe bidder $i$'s cost parameter
$r_{i}$, and with probability $1-q$ he does not observe the cost
parameter. In case a misreport is detected, each bidder $i$ is punished
by at most the amount $P_{i}\leq0$.

Let $H_{i}(s_{i},s_{i}'|r_{i})=\Phi_{i}(r_{i},s_{i})-\Phi_{i}(r_{i},s'_{i}|r_{i},s_{i})$
denote the pre-fee utility gap when a true type $(r_{i},s_{i})$ reports
$(r_{i},s'_{i})$. Define the chain-closure fee for each bidder $i$
as follows:
\[
\tau_{i}^{cc}(r_{i},s_{i})=\inf_{\omega\in\Omega(s_{i})}\left(\Phi_{i}(r_{i},\omega_{1})+\sum_{k=2}^{l(\omega)}H_{i}(\omega_{k},\omega_{k-1})\right)
\]
for any $\theta_{i}=(r_{i},s_{i})\in[0,1]^{2}$. For each fixed $r_{i}$,
the fee scheme $\tau_{i}^{cc}(r_{i},.)$ satisfies stage-$1$ IR and
IC constraints, and maximizes the amount of fee that can be collected
among the feasible schemes. Suppose that bidder $i$ with true type
$(r_{i},s_{i})$ contemplates reporting $(r'_{i},s_{i})$. Then, the
auditing deterrence requires 
\[
\Phi_{i}(r_{i},s_{i})-\tau_{i}^{cc}(r_{i},s_{i})\geq(1-q(r'_{i}))[\Phi_{i}(r_{i},s_{i})-\tau_{i}^{cc}(r'_{i},s_{i})]+q(r'_{i})P_{i}.
\]
 Rearranging the terms gives us the inequality that 
\[
q(r'_{i})\geq\frac{\tau_{i}^{cc}(r_{i},s_{i})-\tau_{i}^{cc}(r'_{i},s_{i})}{\Phi_{i}(r_{i},s_{i})-\tau_{i}^{cc}(r'_{i},s_{i})-P_{i}}.
\]
If $\tau_{i}^{cc}(r'_{i},s_{i})\geq\tau_{i}^{cc}(r_{i},s_{i})$, then
the inequality trivially holds; the fees already establish the deterrence,
and so there is no need for costly auditing. If $\tau_{i}^{cc}(r'_{i},s_{i})<\tau_{i}^{cc}(r_{i},s_{i})$,
then the auditing probability should be at least equal to the right
hand side. Choosing the minimum detection probability deters all possible
misreports for every $s$:
\[
q_{P_{i}}^{*}(r'_{i})=\sup_{r_{i},s_{i}:\tau_{i}^{cc}(r'_{i},s_{i})\leq\tau_{i}^{cc}(r_{i},s_{i})}\left[\frac{\tau_{i}^{cc}(r_{i},s_{i})-\tau_{i}^{cc}(r'_{i},s_{i})}{\Phi_{i}(r_{i},s_{i})-\tau_{i}^{cc}(r'_{i},s_{i})-P_{i}}\right]^{+}.
\]
By stage-$1$ IR and $P_{i}\leq0$, we have $q_{P_{i}}^{*}(r'_{i})\in[0,1]$.
Once again the higher $|P_{i}|$ is, the lower $q_{P_{i}}^{*}(r_{i})$
is for each $r_{i}$, and so the lower auditing cost $k_{r}(q_{P_{i}}^{*}(r_{i}))$
will be. Suppose that the auditing costs are low enough that it is
optimal for the seller to audit. In that case, if $P=(P_{1},...,P_{n})$
is the profile of maximum punishments that can be imposed on bidders,
then in equilibrium the maximum expected revenue under auditing will
be $E_{\theta}[\sum_{i\in N}(x_{i}^{e}(t,\sigma^{*}(\theta))+\tau_{i}^{cc}(\theta_{i}))]-E_{\theta}[\sum_{i\in N}k_{r}(q_{P_{i}}^{*}(\theta_{i1}))]$.

\paragraph*{No auditing:}

Suppose that the seller cannot audit bidder experiments or costs.
In this case, the seller cannot enforce stage-$1$ IC constraint in
the cost-scale parameter $r$ or mean-value parameter $s$. In other
words, any $r$ or $s$ dependent stage-$1$ fee is not incentive
compatible, and so stage-$1$ fees should be flat in both $r$ and
$s$. Since we still have a potential game, the optimal assignments
of experiments can be found by considering the maximizer of the potential
function in $D^{n}$:
\[
\max_{\sigma\in D^{n}}W^{e}(\sigma)=\max_{\sigma\in D^{n}}\left(E_{\theta}[\sum_{k\in N}t_{k}p_{k}^{e}(t,\sigma(\theta))]-E_{\theta}[\sum_{k\in N}c_{\theta_{k}}(\sigma_{k}(\theta_{k}))]\right).
\]
Then the total expected payment at stage-$2$ will be $x^{na}=E_{\theta}[\sum_{i\in N}(x_{i}^{e}(t,\sigma^{na}(\theta))]$
where $\sigma^{na}\in D^{n}$ is the maximizer of the potential function
when there is no audit, and so the bidders can choose any experiment
without any center restriction. In this case, the highest flat fee
that the seller can enforce $\tau_{i}^{na}$ should satisfy stage-$1$
IR constraint. As such, $\tau_{i}^{na}$ should satisfy $\tau_{i}^{na}=\essinf_{(r_{i},s_{i})}\Phi_{i}^{na}(r_{i},s_{i})$
where $\Phi_{i}^{na}(r_{i},s_{i})=\max_{f_{i}\in D}\left(u_{i}^{e}(f_{i}|\sigma_{-i}^{na})-c_{r_{i},s_{i}}(f_{i})\right)$.
Thus, the total expected revenue for the seller will be $x^{na}+\sum_{i\in N}\tau_{i}^{na}$.

\section{Discussion}

Our two-stage mechanism---VCG at stage-$2$ combined with a stage-$1$
experiment and fee schedules (where fees are constructed via the shortest-path
(i.e., the chain-closure) solution to the $s$-IC/IR envelope)---differs
from classic optimal auction models in two intertwined ways: (i) it
preserves efficiency ex-post (VCG) while extracting rents ex-ante
through a fee that internalizes information choices, and (ii) it treats
the stage-$1$ information-acquisition game as a Bayesian exact potential
game, so the equilibrium experiment profile maximizes expected welfare.
Because bidders are ex-ante symmetric, the stage-$1$ experiment assignment
rule and the fee schedule can be taken symmetric across bidders without
loss, which greatly simplifies both characterization and implementation.

Relative to \citet{Myerson81} and \citet{MW82}, who characterize
revenue-optimal allocation/pricing with fixed information, our two-stage
mechanism keeps the allocation rule efficient (VCG) rather than ironing
virtual values. The rationale is that with endogenous signals the
informational externalities run through stage-$2$ payoffs; VCG maximizes
the potential (expected surplus minus information costs), so any distortion
of stage-$2$ would reduce the potential and, by the standard revenue
decomposition, reduce revenue. In contrast, \citet{Stegeman96} and
\citet*{CSZ04} introduce participation/inspection fees and show how
entry or inspection can be screened; our fee, however, is not a uniform
entry or inspection fee but the shortest-path envelope over pairwise
$s$-IC constraints at a given $r$, which achieves maximal rent extraction
subject to truthful centers and stage-$1$ IR.

Within the auction design with endogenous information literature,
mechanisms typically either (a) restrict or recommend information
structures, or (b) embed precision choices into the objective. Our
approach is closest in spirit to \citet{BV02}\textquoteright s efficiency
with costly information and to \citet{BP07}\textquoteright s focus
on information structures, but it leans on two design commitments:
(1) stage-$2$ remain VCG precisely because the stage-$1$ signal
game is a Bayesian exact potential game, so welfare-maximization at
stage-$2$ aligns with equilibrium information choices; and (2) rent
extraction occurs entirely via a stage-$1$ fee, pinned down by the
chain-closure construction that aggregates all local $s$-deviation
bounds into a global maximum-revenue, $s$-IC-and-IR-feasible tariff.
Compared to \citet{Shi12}) and \citet*{GMSZ21}, which allow allocation
distortions to influence information incentives, our two-stage mechanism
isolates information provision from allocation by keeping VCG at stage-$2$
and using the fee to implement the efficient information profile while
minimizing information rents.

There are other related work which treat information as a design primitive
but differ on who controls it, what is optimized, and whether allocation
is distorted. \citet{Mensch22} studies a monopoly screening problem
with limited attention: buyers choose attention and the seller uses
menus (attributes/prices) to separate types, accepting distortions
to extract surplus. \citet{ES07} analyze competitive auctions where
the seller designs public disclosure/handicaps to raise revenue, often
tilting competition and sacrificing efficiency when profitable. By
contrast, our IPV auction keeps allocation efficient by committing
to VCG at stage-$2$ and prices bidder-chosen information ex ante
via a stage-$1$ experiment/fee menu (built from an s-IC/IR shortest-path
envelope). This makes the pre-auction learning game an exact potential
game, implements efficient information acquisition, and lets the seller
achieve optimal revenue without allocation distortions.

Finally, \citet{Persico00} studies signal precision choice at a cost
in first- and second-price auctions within a model of affiliated values;
\citet{PG25} take the second-price allocation rule as given and characterize
how competition shapes how much and what type of information bidders
seek (e.g., learning about own values or rival values). By contrast,
our objective is to find an optimal mechanism which maximizes the
revenue and/or welfare. Our two-stage mechanism designs the information
stage (and fees) so that competition remains effective and revenue-optimal
under IPV with endogenous learning. The symmetry of the environment
implies a symmetric fee/experiment mapping across bidders, and the
potential-game structure implies that the induced information profile
is socially optimal among feasible signals. In short, our two-stage
mechanism sits at the intersection of efficient allocation (VCG),
optimal information provision (through the potential), and maximal
rent extraction (via chain-closure), offering a clean benchmark against
which the classic revenue-maximizing distortions or inspection-fee
designs can be contrasted.

\section{Conclusion}

This paper revisits single-object auctions under independent private
values when bidders can flexibly---but at a cost---improve what
they know about their own valuations before bidding. We contrasted
the welfare benchmark, which implements efficient allocation in VCG/second-price
environments, with the revenue benchmark characterized by virtual
values and reserves. We then asked how endogenous information acquisition
reshapes this classic dichotomy. Our main result shows that, once
the seller commits to a two-stage design in which experiments are
registered up front and the stage-$2$ allocation rule is VCG, the
incentives to learn line up with efficiency, and the seller can implement
a transfer scheme that is incentive compatible and individually rational
while also attaining the Myerson revenue objective. Intuitively, fixing
VCG at stage-$2$ turns the pre-auction information game into an exact
potential game, so equilibrium learning maximizes expected welfare;
the stage-$1$ fee schedule then extracts the value generated by better
selection without needing to condition on unverifiable cost scales,
relying instead on verifiable features of the registered experiments.
This delivers a unification: in our environment, the divergence between
revenue and welfare that arises from exclusion, asymmetries, and reserves
in static models can be neutralized when learning is properly internalized
at stage-$1$.

These findings carry practical and theoretical implications. Practically,
they suggest that when pre-auction diligence is salient---as in spectrum,
mineral rights, or complex procurement---designers should evaluate
formats not only by their static allocation rules but by the learning
incentives they embed; a simple commitment to VCG in stage-$2$ coupled
with a transparent, experiment-based registration fee can simultaneously
deliver high revenue and efficiency. Theoretically, our results refine
the scope of revenue equivalence: formats that are revenue-equivalent
in the no-learning benchmark need not remain so once learning is endogenous,
but an appropriate two-stage commitment can restore alignment by making
the induced allocation rule format-invariant through efficient information
acquisition. Several extensions merit further work: robustness to
richer frictions in verifying experiments; budget or liquidity constraints;
risk aversion; correlated or affiliated signals beyond IPV; dynamic
entry and disclosure; and optimal auditing technologies. Each of these
can interact with learning incentives and reserves in ways that either
reinforce or erode the alignment we document, offering a roadmap for
future research at the intersection of information design and auction
theory.

\smallskip{}
\bibliographystyle{ecta}
\bibliography{oda}

\end{document}